\documentclass{llncs}
\usepackage{amsmath,amsfonts,amssymb,stmaryrd,cmll}
\usepackage{graphicx}
\usepackage[nohug,midshaft]{diagrams}
\usepackage{tikz}
\usepackage[english]{babel}
\DeclareMathOperator{\ob}{ob}

\newarrow{Eqto}=====

\newcommand{\intro}[1]{\emph{#1}}
\newcommand{\bigcat}[1]{\mathsf{#1}}
\newcommand{\cat}[1]{\mathbf{#1}}
\newcommand{\fonc}[1]{\mathsf{#1}}

\newcommand{\Formulae}{\mathcal{F}}
\newcommand{\commentthis}[1]{}

\newcommand{\ie}{i.e.,}
\newcommand{\eg}{e.g.,}

\newcommand{\bsig}{\mathcal{K}} 
\newcommand{\bbig}{\cat{Bbg}} 
\newcommand{\T}{\fonc{T}} 
\newcommand{\theory}{\mathcal{T}}
\newcommand{\sig}{\Sigma}
\newcommand{\A}{\mathcal{A}}
\newcommand{\C}{\mathcal{C}}

\newcommand{\X}{\mathcal{X}}

\newcommand{\Sig}{\bigcat{SMCSig}} 

\newcommand{\SMCCat}{\bigcat{SMCCat}}

\newcommand{\Nat}{\mathbb{N}} 
\newcommand{\smc}{\textsc{smc}}
\newcommand{\imll}{\textsc{imll}}

\newcommand{\loc}{\mathit{loc}}
\newcommand{\prnt}{\mathit{prnt}}
\newcommand{\link}{\mathit{link}}
\newcommand{\ctrl}{\mathit{ctrl}}

\newcommand{\zero}{\mathbf{0}}
\newcommand{\un}{1}
\newcommand{\paral}{\mid}

\newcommand{\transl}[1]{\llbracket #1 \rrbracket}

\newcommand{\iso}{\cong}
\newcommand{\impll}{\multimap}
\newcommand{\tens}{\otimes}
\newcommand{\bigtens}{\bigotimes}
\newcommand{\Bigtens}[1]{\displaystyle{\bigtens_{#1}}}

\newcommand{\card}[1]{|#1|}

\newcommand{\ens}[1]{\{ #1 \}}

\newcommand{\alt}{\mathrel{|}}
\newcommand{\name}[1]{\ulcorner #1 \urcorner}

\title{Binding bigraphs\\as symmetric monoidal closed theories}

\author{Tom Hirschowitz\inst{1} \and Aur\'elien Pardon\inst{2}}
\institute{CNRS, Universit\'e de Savoie
  \email{tom.hirschowitz@univ-savoie.fr} \and ENS Lyon
  \email{aurelien.pardon@ens-lyon.fr}}

\begin{document}
\maketitle

\begin{abstract}

  We reconstruct Milner's~\cite{Milner:bigraphs} category of abstract
  binding bigraphs $\bbig (\bsig)$ over a signature $\bsig$ as the
  free (or initial) \emph{symmetric monoidal closed} (\smc) category $S
  (\theory_\bsig)$ generated by a derived theory $\theory_\bsig$.  The
  morphisms of $S(\theory_\bsig)$ are essentially proof nets from the
  Intuitionistic Multiplicative fragment (\imll) of Linear
  Logic~\cite{ll}.

  Formally, we construct a faithful, essentially injective on objects
  functor $\bbig (\bsig) \to S(\theory_\bsig)$, which is surjective on
  closed bigraphs (i.e., bigraphs without free names or sites).  The
  functor is not full, which we view as a gain in modularity: we
  maintain the scoping discipline for whole programs (bound names
  never escape their scope) but allow more program fragments,
  including a large class of binding contexts, thanks to richer
  interfaces.  Possible applications include bigraphical programming
  languages~\cite{Damgaard:Matching} and Rathke and Soboci\'{n}ski's
  derived labelled transition systems~\cite{modularLTS}.
  \end{abstract}

\section{Overview}
Milner's (binding) bigraphs~\cite{Milner:bigraphs} are a general
framework for reasoning about distributed and concurrent programming
languages, designed to encompass both the
$\pi$-calculus~\cite{Milner:pi} and the Ambient
calculus~\cite{Ambients}.  We are here only concerned with bigraphical
syntax: given what we call a \emph{bigraphical signature} $\bsig$,
Milner constructs a \emph{pre-category}, and then a category $\bbig
(\bsig)$, whose objects are \emph{bigraphical interfaces}, and whose
morphisms are bigraphs.

Its main features are (1) the presence of \emph{relative pushouts}
(RPOs) in the pre-category, which makes it well-behaved w.r.t.\
bisimulations, and that (2) in both the pre-category and the category,
the so-called \emph{structural} equations become equalities. Examples
of the latter are, e.g., in $\pi$ and Ambients, renaming of bound
variables, associativity and commutativity of parallel composition, or
scope extrusion for restricted names. Also, bigraphs follow a scoping
discipline ensuring that, roughly, bound variables never escape their
scope.

In this paper, we reconstruct bigraphs using standard algebraic tools.
To explain this, let us quickly review the notion of a
\emph{many-sorted algebraic theory}, which is central in universal
algebra. It is specified by first giving a \emph{signature}---a set of
sorts $X$ and a set $\sig$ of operations with arities---together with
a set of \emph{equations} over that signature. For example, the theory
for monoids is specified by taking only one sort $x$, and operations
$m: x \times x \to x$ and $e: 1 \to x$, together with the usual
associativity and unitality equations. We may equally well view this
signature as given by a graph
\begin{equation}
 \begin{diagram}
  (x \times x) & \rTo^{m} & x & \lTo^{e} & 1
\end{diagram}\label{eq:monoid}
\end{equation}
with vertices being the objects of the free category with (strict)
finite products generated by $X$.

In this paper, we use the same kind of theories, but replacing from
the start finite products with \smc\ structure.  \smc{} structure is
the categorical counterpart of \imll{}~\cite{ll,Trimble:phd}: through
the Curry-Howard-Lambek correspondence, an \smc{} signature amounts to
a set of \imll{} axioms, and the free \smc{} category $S (\Sigma)$
over a signature $\Sigma$ has as morphisms \imll{} proofs under the
corresponding axioms, modulo cut elimination.  We use an isomorphic
presentation of $S (\Sigma)$, essentially due to
Trimble~\cite{Trimble:phd}, in which morphisms are very much like
\imll{} \emph{proof nets}~\cite{ll}: they are kind of graphs, whose
\emph{correctness} is checked by (a mild generalisation of) the
well-known Danos-Regnier criterion~\cite{Danos:mll}.

Here, we translate any bigraphical signature $\bsig$ into an \smc{}
theory $\theory_\bsig$, and consider the free \smc\ category $S
(\theory_\bsig)$ generated by $\theory_\bsig$ as an alternative
category of bigraphs over $\bsig$.  To compare $S (\theory_\bsig)$
with Milner's $\bbig (\bsig)$, we construct a faithful functor $\T:
\bbig (\bsig) \to S (\theory_\bsig)$. This functor is moreover
essentially injective on objects (i.e., two objects with the same
image are isomorphic), so that it is essentially an embedding.

However, our functor $\T$ is not full: even between bigraphical
interfaces, $S (\theory_\bsig)$ contains morphisms which would be
ill-scoped according to Milner's scope rule. Neither is $\T$
surjective on objects: $S (\theory_\bsig)$ has many more objects
beyond bigraphical interfaces. This could be perceived as negative at
first sight, but it actually represents a gain in
modularity. Informally, even though our category is much more
versatile, it still prevents bound variables to escape their
scope. More formally, our functor is full on whole programs, i.e.,
bigraphs with no sites nor open names. Namely, it induces an
isomorphism on closed terms, i.e., of hom-sets $$S (\theory_\bsig) (I,
t) \iso \bbig (\bsig) (I, t),$$ where $I$ is the unit of tensor
product, and $t$ is a particular object representing terms\footnote{We
  cheat a little here, see the actual result
  Lemma~\ref{lemma:closed}.}.  So, our additional interfaces and
morphisms allow just as many whole programs as Milner's, but more
program fragments. Notably, it contains both the equivalent of terms
and a kind of multi-hole, higher-order, binding contexts, all
cohabiting happily.

In passing, our functor fully elucidates the status of so-called
\emph{edges} in bigraphs: we translate differently \emph{free} edges
(used for name restriction, much like $\nu$ in the $\pi$-calculus) and
\emph{bound} edges (used for linking so-called \emph{binding ports} to
their peers). In the former case, we translate the edge into a $\nu$
node (which may also be understood as representing the private name in
question); in the latter case, we simply remove the edge, and rely on
our use of directed graphs to represent the flow from the binding port
to its peers.

Finally, our algebraic approach directly extends to what
Debois~\cite{Debois:phd} calls discrete sortings, which amount to a
many-sorted variant of bigraphs. Given a set $X$ of sorts, instead of
only two sorts $t$ and $v$ for terms and variables (or names), one
starts with sorts $t_x$ and $v_x$ for each $x \in X$, which directly
allows to specify a typed signature.

\paragraph{Future work}
A possible application of our work is as an alternative representation
for bigraphical programming
languages~\cite{Damgaard:Matching}.  For example,
on the graphical side, the extensive litterature on efficient
correctness criteria for proof nets applies directly.  On the
syntactic side, \imll{} proofs provide an essentially algebraic
representation, i.e., one avoiding the use of variable binding and the
associated
trickery~\cite{PfenningElliott:hoas,PittsAM:newaas,hofmann:presheaf}.

Another possible application is as a handy foundation for Rathke and
\linebreak Soboci\'{n}ski's~\cite{modularLTS} work on deconstructing labelled
transition systems, which involves second-order (binding) contexts.

As to future research directions, we here only handle abstract
bigraphs, which do not have RPOs. We thus should generalise our
approach to deal with \emph{concrete} bigraphs, be it in the form of
Milner's original pre-category or in Sassone and Soboci\'{n}ski's
\emph{G-categories}~\cite{Sobocinski:grpos}, and then try to construct
the needed RPOs (or GRPOs).  

Another natural research direction from this paper concerns the
dynamics of bigraphs.  Our hope is that Bruni et
al.'s~\cite{Bruni:ccdc} very modular approach to dynamics may be
revived, and work better with \smc\ structure than with cartesian
closed structure. Specifically, with \smc\ structure, there is no
duplication at the static level, which might simplify matters.

\paragraph{Related work}
The construction of the free \smc\ category generated by an \smc\
theory is essentially due to Trimble~\cite{Trimble:phd}, followed by
others~\cite{Tan:phd,Hughes:freestar,Lamarche:nets}.  The construction
we use is a variant of Hughes'~\cite{Hughes:freestar} construction,
defined in our joint work with Richard Garner~\cite{GHP}.  It was
known that \smc\ (or cartesian closed) structure precisely represents
various kinds of variable binding (see, e.g.,
Barber~\cite{Barber:action}).  So we do not claim originality for the
use of \smc{} structure to recover bigraphs.

Sassone and Soboci\'{n}ski~\cite{Sobocinski:graphslics} share our goal of
categorically reconstructing bigraphs. They obtain very satisfactory
results for pure bigraphs as a bicategory of cospans over a particular
category of graphs. But their approach still has to scale up to deal
with binding.

Damgaard and Birkedal~\cite{Birkedal} axiomatise the category of
bigraphs as an equational theory over a term language with variable
binding. Our work may be seen as an essentially algebraic counterpart
of theirs (which relies on variable binding and the trickery evoked
above).  Moreover, they do not recognise the special status of bound
edges, so our construction is not a mere reformulation of theirs in
categorical terms.

Milner~\cite{Milner:bigraphs2} and Debois~\cite{Debois:phd}
propose other extensions of Milner's original scope condition, the
latter subsuming the former~\cite[Section~6.4]{Debois:phd}.  Debois
sees binding as a \emph{sorting} over the category of pure bigraphs
generated by $U(\bsig)$, where $U$ forgets the binding information.
His construction, starting from a `no bound name escapes its scope'
predicate, is reminiscent from Hyland and Tan's \emph{double glueing}
construction~\cite{Tan:phd}. However, his construction is not known to
satisfy any universal property, and an efficient (e.g., algebraic)
representation of it still has to be found. 

Finally, Grohmann and Miculan propose \emph{directed
  bigraphs}~\cite{grohmann} as a more flexible framework for
bigraphs. Hildebrandt~\cite{Hildebrandt:choco} proposes an algebraic
approach to them, using \emph{compact closed} categories instead of
\smc{} categories. This line of work does not handle scope directly:
one has to resort to a sorting in the above sense, with the same
inconveniences.

\paragraph{Summary of contributions}
In view of the above, our contribution thus mainly lies in working out
the details of the encoding of bigraphs as \smc{} theories, plus
unraveling and simplifying the status of edges. Our task is made
easier by Hughes' economic presentation of the free \smc{} category.

Furthermore, previous work using \smc{} structure as a representation
for binding mostly lead to conservative extensions, i.e., full and
faithful functors. Our work emphasises that Milner's \emph{ad hoc}
condition leaves room for generalisation, and provides a new,
canonical condition, which benefits from efficient criteria from
linear logical literature.

\paragraph{Structure of the paper}
Section~\ref{sec:smctheories} recalls the construction of the free
\smc{} category generated by an \smc\ theory, including a
specialisation to the case where a sort is equipped with a commutative
monoid object structure.  Section~\ref{sec:big} reviews bigraphs,
defining along the way our translation of bigraphical signatures. In
Section~\ref{sec:trans}, we construct our functor from bigraphs to the
corresponding free \smc\ category, and show that it is an isomorphism
on closed terms.

\section{Symmetric monoidal closed theories}\label{sec:smctheories}
This section reviews \smc{} theories and their free models; see our
note~\cite{GHP} for details. The constructions are essentially due to
Trimble~\cite{Trimble:phd}, but reworked using our extensions to
Hughes'~\cite{Hughes:freestar} presentation.

\subsection{Signatures}
We should start our overview with the definition of \smc{} categories:
these are symmetric monoidal categories, i.e., categories with a
tensor product $\tens$ on objects and morphisms, symmetric in the
sense that $A \tens B \iso B \tens A$, such that $(- \tens A)$ has a
right adjoint $(A \impll -)$, for each $A$. We do not give further
details, since we are interested in describing the free such category,
which happens to be easier. Knowing that there is a category $\SMCCat$
of \smc{} categories and strictly structure-preserving functors should
be enough to grasp the following.

As sketched in the introduction, a \emph{signature} $\Sigma$ is given
by a set $X$ of sorts, equipped with a graph whose vertices are
\imll{} formulae over $X$, as defined by the grammar:
$$\begin{array}{rcl@{\hspace*{2cm}}r}
  A, B, \ldots \in \Formulae (X) & ::= & x \alt I \alt A \tens B \alt A \impll B & 
  \mbox{(where $x \in X$).}
\end{array}$$
We think of each edge $A \to B$ of the graph as specifying an operation 
of type $A \to B$.
A morphism of signatures $(X, \Sigma) \to (Y, \Sigma')$ is a 
function $X \rTo^{f} Y$, equipped with a morphism of graphs, whose 
vertex component is ``$\Formulae (f)$'', i.e., it sends 
any formula $A (x_1, \ldots, x_n)$ to $A (f (x_1), \ldots, f (x_n))$.
This defines a category $\Sig$ of signatures.

There is then a forgetful functor $\SMCCat \rTo^{U} \Sig$ sending
each \smc{} category $\C$ to the graph with
\begin{description}
\item[vertices:] formulae in $\Formulae (\ob (\C))$, and
\item[edges $A \to B$:] morphisms $\transl{A} \to \transl{B}$ in $\C$,
  where $\transl{A}$ is defined inductively to send each syntactic
  connective to the corresponding function on $\ob (\C)$.
\end{description}

Trimble~\cite{Trimble:phd} constructs an \smc{} category $S (\Sigma)$
from any signature $\Sigma$, which extends to a functor $\Sig \rTo^{S}
\SMCCat$, left adjoint to $U$: for any \smc\ category $\C$ and natural
transformation $\Sigma \rTo^{f} U (\C)$, there is a unique \smc\
functor $S (X) \rTo^{f^*} \C$ such that $f$ decomposes as
\begin{diagram}[inline]
  X & \rTo^{\eta} & US (X) & \rTo^{U (f^*)} & U(\C).
\end{diagram}

\subsection{The free \smc{} category}
How does $S (\Sigma)$ look like?  Its objects are \imll{} formulae in
$\Formulae (X)$ and morphisms are kind of proof-nets with some
\emph{cells} representing operations.  A morphism from $A$ to $B$ in
$S (\Sigma)$ thus consists of
\begin{itemize}
\item a finite set $C = \ens{c, \ldots}$ of cells labelled by
  operations $\alpha_c \to \beta_c$ in $\Sigma$,
    \item and \emph{wires} connecting the \emph{ports} together.
\end{itemize}
For example in the $\pi$-calculus, the operations $\mathit{get}$ and
$\mathit{send}$ correspond to cells:
\begin{center}
\includegraphics{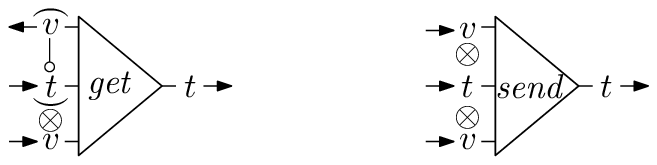}
\end{center}
We always use the flat edge (or base) of the polygon to denote the domain. 
The orientation of a port corresponds to its sign, see below. 

A port is a leaf occurrence (atomic or $I$) in $A$, $B$, or in some
$\alpha_c$, $\beta_c$.
Equivalently, a port is a leaf occurrence in the formula
\begin{equation}
  (A \tens \Bigtens{c \in C} (\alpha_c \impll \beta_c)) \impll B\label{eq:ports}
\end{equation}
(or $A \impll B$ when $C$ is empty).
In a formula, a port is \emph{positive} when it lies to the left of an even
number of $\impll$'s, and \emph{negative} otherwise\footnote{The sign of a port
in $A$ is directly apparent viewing $A$ is a classical LL formula, see the next
paragraph.}.
The sign of a port in a morphism is its sign in the formula~\eqref{eq:ports}.

Wires are oriented, with sources the \emph{negative} ports of the
morphism, and targets in its \emph{positive} ports.  For each sort $x
\in X$, the wires must induce a bijection between ports labelled $x$
(which we call $x$ ports).  A negative $I$ port can be wired to any
positive port.

An example of morphism is presented in Fig.~\ref{freesig:example:ii}, where
$\Sigma$ has sorts $\{t,v\}$ and operations $\mathit{get}$, $\mathit{send}$, $v
\rTo^{c} v \tens v$, $I \rTo^{\nu} v$ and $t \tens t \rTo^{p} t$.
\begin{figure}[ht] \centering
    \includegraphics{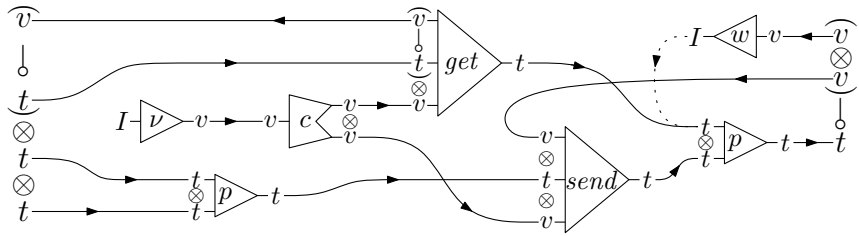}
    \caption{A morphism of $S(\Sigma)$.}
    \label{freesig:example:ii}
\end{figure}

\paragraph{Correctness}
But, crucially, not all such graphs qualify as morphisms of $S
(\Sigma)$: they have to satisfy a correctness criterion essentially
due to Danos-Regnier~\cite{Danos:mll}, which goes as follows.
The formula~\eqref{eq:ports} may be written using the connectives
of \emph{classical} linear logic, defined by the grammar:
$$\begin{array}{rcl@{\quad \mid \quad}l@{\quad \mid \quad}l}
A, B, \ldots & ::= & x    &  \un &  A \otimes B \\
&  \mid   & x^\bot & \bot & A \parr B.
\end{array}$$
We have removed $A \impll B$, now encoded as $A^\bot \parr B$; some
classical formulae are not expressible in \imll{}, such as $\bot$, or
$x \parr x$. The de Morgan dual $A^\bot$ of $A$ is defined as usual.

A \emph{switching} of a classical formula is its abstract syntax tree,
where exactly one argument edge of each $\parr$ has been removed. A
\emph{switching} of a candidate morphism $f$ is a graph obtained by
gluing along ports the wires of $f$ with a switching of the
formula~\eqref{eq:ports}. The candidate is then \emph{correct} iff all
its switchings are acyclic and connected.

Equivalently, since~\eqref{eq:ports} is ultimately a $\parr$, one may
separately glue $f$ with switchings of $A^\bot$, $B$, and each
$\alpha_c \tens \beta_c^\bot$. Notably, a cell of type $A \to x$ for
some sort $x$ is switched by connecting $x$ to a switching of $A$.

In Fig.~\ref{freesig:super_example} is pictured one switching (among 64) of
the graph underlying the morphism in Fig.~\ref{freesig:example:ii}. 
The displayed connectives are those of the formula~\eqref{eq:ports} translated
into classical linear logic. 
\begin{figure}[ht] \centering
    \includegraphics{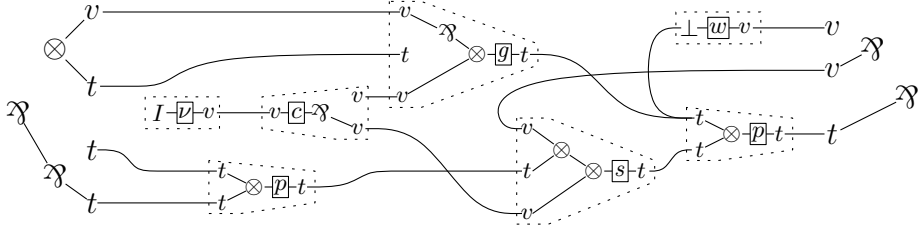}
    \caption{A connected and acyclic switching of the morphism in Fig.~\ref{freesig:example:ii}.}
    \label{freesig:super_example}
\end{figure}

\paragraph{Rewiring} 
Finally, morphisms are quotiented by Trimble \emph{rewiring}: a
morphism \emph{rewires} to another by changing the target of an edge
from some negative $I$ port, as soon as this preserves
correctness. Rewiring is the smallest equivalence relation generated
by this relation.  
In Fig.~\ref{freesig:example:ii}, the dotted wire starting from the
negative $I$ port can be rewired to any positive port because all switchings
would remain trees. 

\subsection{Theories}\label{subsec:co:monoids}
That gives the construction for signatures. We now extend it to \smc\
theories: define a theory $\theory$ to be given by a signature $\sig$,
together with a set $E_{A,B}$ of equations between morphisms in $S
(\sig) (A, B)$, for each \imll\ formulae $A, B$.  The free \smc\
category $S (\theory)$ generated by such a theory is defined in our
note~\cite{GHP} to be the quotient of $S (\sig)$ by the
equations. Constructing $S (\theory)$ graphically is more direct than
could have been feared: we first define the binary predicate $f_1 \sim
f_2$ relating two morphisms
$C \pile{\rTo^{f_1, f_2} \\
  \rTo} D$ in $S (\sig)$ as soon as each $f_i$ decomposes as
\begin{diagram}
  C & \rTo^{\iso} & I \tens C & \rTo^{\name{g_i} \tens C} & (A \impll
  B) \tens C & \rTo^{f} & D
\end{diagram}
with a common $f$, with $(g_1, g_2) \in E_{A, B}$, and where $\name{g}$
is the currying of $g$. Then, we take the smallest generated
equivalence relation, prove it stable under composition, and quotient
$S (\sig)$ accordingly, which yields the free \smc\ category $S
(\theory)$ generated by the theory $\theory = (\sig, E)$.

\paragraph{Commutative monoid objects}
We finally slightly extend the results of our note~\cite{GHP} to
better handle the special case of commutative monoids objects. This
will be useful in our translation of bigraphs, where the sort $t$ of
terms has a commutative monoid structure given by parallel composition
and $\zero$. Assume a theory $(\sig, E)$ where a sort $t$ is equipped
with two operations $m$ and $e$ as in~(\ref{eq:monoid}), with
equations making it into a commutative monoid object ($m$ is
associative and commutative, $e$ is its unit).  Further assume that
$m$ and $e$ do not occur in other equations.

Let $\sig'$ be the result of removing the operations $m$ and $e$ in
$\sig$.  We define a relaxed version of our morphisms where each
negative $t$ port is connected to a positive one,
but not necessarily bijectively.  This defines a category isomorphic
to $S(\sig)$, in which the operations $m$ and $e$ are built into
the linking.  The isomorphism is pictured in Fig.~\ref{monoid_to_c}.
\begin{figure}[ht] \centering
    \includegraphics{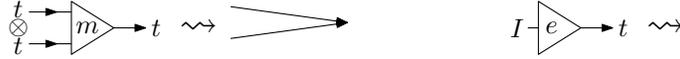}
    \caption{Contracting $m$ cells and deleting $e$ cells.}
    \label{monoid_to_c}
\end{figure}

\section{Binding bigraphs and the translation of signatures}\label{sec:big}
We now proceed to recall some definitions from Milner~\cite{Milner:bigraphs}, along which we give our translation of
bigraphical signatures $\bsig$ into \smc\ theories $\theory_\bsig$. We
then turn to our translation from the corresponding category of
bigraphs to the free model $S (\theory_\bsig)$.

\subsection{Signatures}
\label{bigraph:sig}
\begin{definition}
  A \intro{bigraphical (binding) signature} is a 4-uple $(
  \bsig,B,F,\A )$ where $\bsig$ is a set of \intro{controls}, $B,F:
  \bsig \to \Nat$ are maps providing a \emph{binding} and a
  \emph{free} arity for each control and $\A \subseteq \bsig$ is a set
  of \emph{atomic} controls.
\end{definition}
We fix such a bigraphical signature $\bsig$ for the rest of the paper. 
This signature can be translated into a \smc\ signature $\sig_\bsig$ over two
sorts $\{t,v\}$, standing for terms and variables (or names).
It consists of the following \emph{structural} operations, accounting for the
built-in structure of bigraphs:
$$
\begin{diagram}[nohug,width=1.2cm,height=.3cm]
    t \tens t & \rTo^{\mid} & t & \lTo^{\zero} & I \\
    & & & & & & I & \rTo^{\nu} & v \\
    v \tens v & \lTo^{c} & v & \rTo^{w} & I \\
\end{diagram}
$$
plus, for all controls $k$, a \emph{logical} operation
$$ \begin{diagram}[inline,width=4em]
(v^{\tens B(k)} \impll x) \tens v^{\tens F(k)} & \rTo^{\quad K_k} & t 
\end{diagram} 
$$ 
where $x = I$ if $k$ is atomic and $x = t$ otherwise.

We call $\theory_\bsig$ the theory consisting
of the operations in $\sig_\bsig$, with the equations making
\begin{itemize}
\item $(t, \paral, \zero)$ into a commutative monoid object, 
\item $(v, c, w)$ into a cocommutative comonoid object ($c$ is
  coassociative, cocommutative, and $w$ is its unit), and
\item $\nu$ and $w$ annihilate each other, as in
	\begin{center}
	    \includegraphics{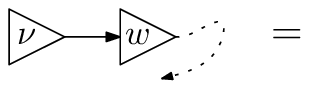} $\cdot$
	\end{center}
\end{itemize}
We now proceed to describe the category $\bbig(\bsig)$ of
\emph{abstract binding bigraphs} over $\bsig$, which we relate in
Section~\ref{sec:trans} to the free model $S (\theory_\bsig)$ of
$\theory_\bsig$.

\subsection{Interfaces}
We assume an infinite and totally ordered set $\X$ of \emph{names}. 

\begin{definition}
  A \intro{bigraphical (binding) interface} is a triple $(
  n,X,\loc )$ where $n$ is a finite ordinal, $X$ a finite set of
  names and $\loc: X \to n + \{ \bot \}$ a function
  called \intro{locality map}.
\end{definition}
A name $x$ is said \emph{global} if $\loc(x) = \bot$ and \emph{local} or
\emph{located at $i$} when $\loc(x) = i \in n$.

Bigraphical interfaces are the objects of the category $\bbig(\bsig)$.
We define a function $\T$ from these objects to \imll\ formulas, \ie\
objects of $S (\theory_\bsig)$, by:
\begin{equation}
\begin{array}{c c c c c}
    \T & : & ( n,X,\loc) & \mapsto & v^{\tens n_g} \impll \Bigtens{i \in n} ( v^{\tens n_i} \impll t) 
  \end{array}\label{eq:T:ob}
\end{equation}
  where $n_g = \card{\loc^{-1}(\bot)}$ and for all $i \in n$, $n_i =
  \card{\loc^{-1}(i)}$.
  The ordering on $\X$ induces a bijection between $X$ and  $v$ leaves in the formula. 

In \cite{Milner:bigraphs2}, Milner presents a slight generalisation of binding
bigraphs, where names have multiple locality.
Some interfaces cannot be simply translated into \imll\ formulas as before,
\eg\ if $x$ is located in $0$ and $1$ and $y$ in $1$ and $2$, this dependency
cannot be expressed directly in an \imll\ formula.

\subsection{Place graph}
Let $n$ and $m$ be two finite ordinals. 
\begin{definition}
    A \intro{place graph} $( V,\ctrl,\prnt ) : n \to m$ is a pair where:
    \begin{itemize}
	\item $V$ is a finite set of \emph{nodes},
	\item $\ctrl: V \to \bsig$ is a function called \emph{control map} and
	\item $\prnt: n + V \to V + m$ is an acyclic function
	    called \emph{parent map} whose image does not contain any
	    atomic node.
    \end{itemize}
\end{definition}
The ordinals $n$ and $m$ index respectively the \emph{sites} and \emph{roots}.
A node is said \emph{barren} if it has no preimage under the parent map (atomic
nodes are thus barren). 

The relation $\prec$ over sites, roots and nodes defined by:
$$ x \prec y \iff \exists k > 0\,,\ \prnt^k(x) = y $$
is a (strict) partial order. 
The maximal elements of $\prec$ are the roots; the minimal elements are the
barren nodes (including atomic nodes) and the sites. 

\subsection{Link graph}
Let $X$ and $Y$ be two finite sets of names. 
\begin{definition}
    A \intro{link graph} $( V,E,\ctrl,\link ) : X \to Y$ is a tuple
    where:
    \begin{itemize}
	\item $V$ is a finite set of \emph{nodes},
	\item $E$ is a finite set of \emph{edges},
	\item $\ctrl: V \to \bsig$ is a \emph{control map} and
	\item $\link: P + X \to E + Y$ is a function called the \emph{link
	    map} 
    \end{itemize}
    with $P$ being the set of \emph{ports}, \ie\ the coproduct of binding ports
    defined by $P_B = \coprod_{v \in V} B(\ctrl(v))$ and free ports $P_F =
    \coprod_{v \in V} F(\ctrl(v))$. 
    Moreover, $\link$ must satisfy the \emph{binding rule}:
    \begin{quote}
	For all binding ports $p \in P_B$, $\link(p) \not \in Y$. 
    \end{quote}
\end{definition}

This binding rule is not mentionned in the original
paper~\cite{Milner:bigraphs} about bigraphs whereas it is mandatory
for the scoping discipline to be stable under composition (it is added
in~\cite{Milner:bigraphs2}). 
Alternatively, we can only require $\link(p)$ to not be a global name of $Y$;
the scope rule (defined below) handles the case of local names of $Y$. 

We define the \emph{binders} of our link
graph to be the local names of $Y$ (located at a root) and the binding
ports (located at a node) $P_B$.
Two distinct \emph{points} (\ie\, two elements of $P + X$) $x$ and $y$ are
\emph{peers} when $\link(x) = \link(y)$. 
An edge is \emph{idle} when it has no preimage under the link map. 

\subsection{Abstract binding bigraphs}
\label{bigraph:sec:def}
Let $U = ( n,X,\loc )$ and $W = ( m,Y,\loc' )$ be two
bigraphical interfaces. 
\begin{definition}
    A bigraph $G = ( V,E,\ctrl,\prnt,\link ) : U \to W$ is a tuple where: 
    \begin{itemize}
	\item $(V,\ctrl,\prnt): n \to m$ is a place graph,
	\item $(V,E,\ctrl,\link): X \to Y$ is a link graph, 
	\item $G$ satisfies the \emph{scope rule}:
	    \begin{quote}
		If $p$ is a binder located at $w$, then each of its peers is located at
		some $w' \prec w$. 
	    \end{quote}
    \end{itemize}
\end{definition}

\begin{example}
  An example of bigraph, roughly corresponding to the $\pi$-calculus
  context $$\nu x (\bar{x}y.(\square_2 \mid \square_3) \mid
  x(z).\square_1)$$ where a global variable $t$ is not used, is given in Fig.~\ref{bigraph:milner} using a
  representation from~\cite{Milner:bigraphs} and another from~\cite{Hildebrandt:choco}.  Binding (resp. free) names and ports are
  pictured by $\bullet$ (resp. $\circ$).
\end{example}
\begin{figure}[ht!]\centering
    \includegraphics{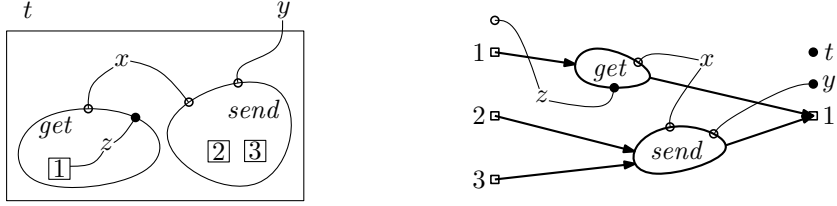}
    \caption{Two bigraphical representations of the same $\pi$-term.}
    \label{bigraph:milner}
\end{figure}

The binding rule ensures that no binding port $p$ is peer of a name in
$Y$, hence $\link (p)$ has to be an edge. Moreover, by acyclicity of
$prnt$, no two binding ports may be peers, hence edges are linked to
at most one binding port.  The set of edges may thus be decomposed
into a set of \emph{free} edges $E_F$ (without binding port) and a set
of \emph{bound} edges $E_B$ in one-to-one correspondence with $P_B$ by
the link map: $E = E_F \uplus E_B \iso E_F + P_B$.

Finally, two bigraphs are \emph{lean-support} equivalent when after discarding
their idle edges, there is an isomorphism between their sets of nodes and edges
preserving the structure. 

\begin{definition}
    The category $\bbig(\bsig)$ of \emph{abstract binding bigraphs over
    $\bsig$} has bigraphical interfaces as objects and lean-support equivalence
    classes of bigraphs as morphisms. 
\end{definition}

The composition of two bigraphs
\begin{diagram}[inline,width=.7cm]
  U_1 & \rTo^{G} & U_2 & \rTo^{G'} & U_3
\end{diagram} is defined by taking the
coproduct of their nodes, edges and control maps and the composition of parent
and link maps (modulo some bijections on sets), forgetting the roots/sites from
$U_2$.
Acyclicity of the parent map, and the binding and scope rules are preserved by
composition.

\section{Translation}\label{sec:trans}

We now want to show how a binding bigraph $G = ( V,E,\ctrl,\prnt,\link
): U \to W$ over $\bsig$ can be translated into a morphism $\T(G):
\T(U) \to \T(W)$ in the free model $S (\theory_\bsig)$ of the \smc\
theory $\theory_\bsig$.  We have already defined $\T$ on objects
in~\eqref{eq:T:ob}. Now, let $U = ( n,X,\loc )$ and $W = ( m,Y,\loc'
)$. We will define the support $C$ of $\T(G)$ as the disjoint union
of:
\begin{itemize}\item a \emph{logical} support $C$ containing a
  $K_k$ cell for every node whose control is $k$ and a $\nu$ cell for
  every free edge in $E_F$, and
\item a \emph{structural} support $C'$ consisting of $c$ and $w$
  cells, which we define below.
\end{itemize}
We then specify the graph $\T (G)$ for each sort in $\ens{t, v}$
separately, and for $I$.  For example, the image by $\T$ of the
bigraph in Fig.~\ref{bigraph:milner} is the morphism in
Fig.~\ref{freesig:example:ii}. 

\subsection{Places}

First, since $(t, |, \zero)$ has a commutative monoid object
structure, the representation of Section~\ref{subsec:co:monoids}
applies: we just have to define a function from negative $t$ ports to
positive ones.  Now, for any set $X$ labeled in formulae, denote by
$X^+_t$ its set of positive $t$ ports, and similarly for
$X^{+,-}_{t,v,I}$. Now, considering each cell $c$ to be labelled by
the formula $\alpha_c \impll \beta_c$, we have:
\begin{itemize}
\item $C^+_t \iso V$, because each type of cell $K_k$ has one positive
  $t$ port,
\item $C^-_t \iso V_{\mathit{na}} \hookrightarrow V$, where
  $V_{\mathit{na}}$ is the set of non-atomic nodes, because there is
  one negative $t$ port for each non-atomic cell,
\item $\T (U)^+_t \iso n$, because for each $i \in n$ there is a
  positive $t$ port in $\T (U)$,
\item similarly, $\T (W)^+_t \iso m$, and finally
\item $\T(W)^-_t \iso \T(U)^+_t \iso \emptyset$. 
\end{itemize}
Our morphism $\T (G)$ is thus defined on the sort $t$ by
the function $f_t$:
\begin{diagram}[width=2cm,height=0.8cm]
  \T (U)^+_t + C^+_t + \T (W)^-_t & 
  \rTo^{\iso} & \T (U)^+_t + C^+_t & \rTo^{\iso} & n + V \\
  \dDashto<{f_t} &&&& \dTo>{\prnt} \\
  \T (U)^-_t + C^-_t + \T (W)^+_t & \lTo^{\iso} & \T (V)^+_t + C^-_t & \lTo^{\iso} & m + V_{\mathit{na}} \text{.}
\end{diagram}

\subsection{Links}
The function $f_v$ for $v$ requires more work, and involves defining
the structural support $C'$.  Recall that the data is the function
$\link: P \uplus X \to E \uplus Y$.

We start with an informal description of $f_v$ based on
Fig.~\ref{bigraph:t_e}, in which bold arrows come from binders.
First, we deal with points sent to edges.  There are two kinds of
edges.  

First, we understand each free edge $e$ as the creation of a fresh
name, and each free point $p$ in $P_F \uplus X$ sent to $e$ as an
occurrence of this free name. Accordingly, $e$ is replaced by its
$\nu$ cell in $C$, and each $p$ becomes a $v$ port in $\T (U)^- +
C^-$.  We hence link the $v$ port of the $\nu$ cell to each
corresponding $p$, through a tree of $c$ and $w$ cells, as depicted in
the bottom row.

Second, we understand each bound edge $e$ as an indirection to its
binding port $p_0 \in P_B$, itself understood as a bound name. We
further understand each free peer $p \in P_F \uplus X$ of $p_0$ as an
occurrence of the bound name.  Accordingly, we completely forget about
$e$, $p_0$ becomes a $v$ port in $C^+$, and each $p$
becomes a $v$ port in $C^- + \T (U)^-$, hence we link $p_0$ to each
corresponding $p$, again through a tree of $c$ and $w$ cells.

Finally, points $p$ not sent to an edge are sent to some name $y \in
Y$.  But each such $p$ becomes a $v$ port in $\T (U)^- + C^-$ and each
such $y$ becomes a $v$ port in $\T (W)^-$, hence we link $y$ to each
$p$, again using $c$ and $w$ cells.  This determines the structural
support $C'$, as well as $f_v$.  Finally, for the $I$ part $f_I$, each
negative $I$ port arises from a structural $w$ cell. But in
the above each cell is generated by \emph{one} $v$ port (the fresh
name or the binder). In the former case, we may safely link our $I$
port to any valid $t$ port. In the latter, the binder occurs to the
left of a $\impll$, whose right-hand side is a $t$ port, to which we
safely link our $I$ port.

\begin{figure*}[t]\centering
    \includegraphics{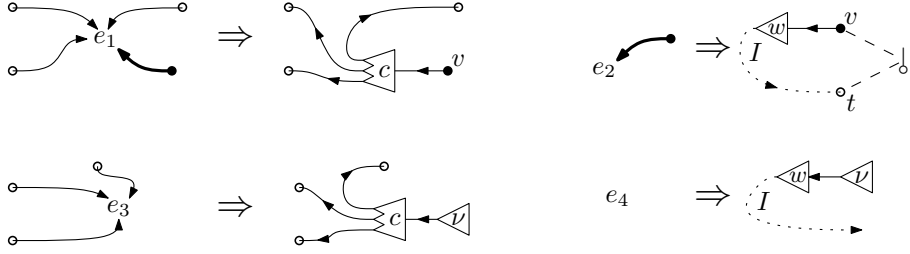}
    \caption{Translation of $\link$.}
    \label{bigraph:t_e}
\end{figure*}

More formally, observe from our translation of signatures and
interfaces, plus the logical support $C$ defined above, that:%
\begin{itemize}
\item each free edge in $E_F$ corresponds to one $\nu$ cell,
  hence to one port in $C^+_v$,
\item each binding port in $P_B$ corresponds to one negative
  occurrence of $v$ in the domain of some cell in $C$, hence to one
  port in $C^+_v$,
\item each local name in $Y$ corresponds  
  one port in $\T (W)^-_v$.
\end{itemize}
Thus, we have an isomorphism $E_F + P_B + Y \iso C^+_v + \T (W)^-_v$.
Similarly, free points in $P_F + X$ correspond to ports in $C^-_v + \T
(U)^-_v$, i.e., $P_F + X \iso C^-_v + \T (U)^-_v$.

We may thus define a first function $\link'$ by:
$$ \begin{diagram}[width=2cm,height=0.8cm]
    C^-_v + \T (U)^-_v & \rTo^{\iso} & P_F + X & \rInto & P_B + P_F + X \\
   \dDashto<{link'} &&&& \dTo>{\link} \\
   C^+_v + \T (W)^-_v & 
   \lTo^{\quad \iso} & E_F + P_B + Y & \lTo^{\quad \iso} & E + Y \text{.}
\end{diagram} $$

We then encode this function by a forest of $c$ and $w$ cells $C'$ (as
pictured in Fig.~\ref{function_to_cwtree}), to obtain a function
$C^+_v + C'^+_{v} + \T (W)^-_v \rTo^{f_v} C^-_v + C'^-_{v} + \T
(U)^-_v$, which qualifies as the $v$ part of our morphism. %
\begin{figure}[ht]\centering
    \includegraphics{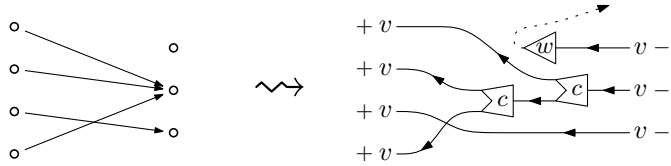}
    \caption{Translation of a function using $w$ and $c$ cells.}
    \label{function_to_cwtree}
\end{figure}%
The rest follows similarly.


This defines a function from bigraphs to candidate morphisms
(respecting domain and codomain).  We now show that it extends to a
functor.
\subsection{The functor}

First, we prove that the image of a bigraph is correct, i.e., is a
proper morphism.

\begin{lemma}\label{lemma:connected}
    All switchings of $\T(G)$ are connected. 
\end{lemma}
\begin{proof}
  Essentially an induction over the place ordering $\prec$.
\end{proof}
The following seems known \cite{Soloviev:imll}:
\begin{lemma}
  Any switching of a morphism in $S(\theory_\bsig)$ is acyclic iff it
  is connected.
\end{lemma}
\begin{proof}[sketch]
    One proves by induction on the domain and codomain formulae that the
    graph induced by the switching has one more vertex than it has edges.
\end{proof}

\begin{lemma}
    The map $\T: \bbig(\bsig) \to S(\theory_\bsig)$ is a functor. 
\end{lemma}
\begin{proof}[sketch]
    The equations of $\theory_\bsig$ defined in Section~\ref{bigraph:sig}
    ensure that $\T$ behaves well w.r.t.\ composition and lean-support
    equivalence. 
\end{proof}

One sees at once that $\T$ is not full. For example, the morphism in
Fig.~\ref{bigraph:iso_1} has no preimage -- any such preimage would
violate the scope rule for bigraphs.
\begin{figure}[!ht] \centering
    \includegraphics{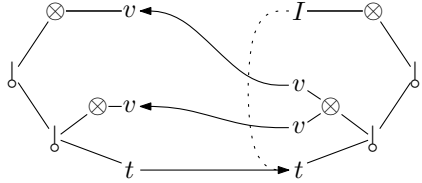}
    \caption{A correct morphism violating the scope rule.}
    \label{bigraph:iso_1}
\end{figure}
This example reflects that it is not necessary to distinguish global and local
variables in a bigraph with only one site. 
Nevertheless, the notion of scope is preserved by $\T$ because closed morphisms
can actually be translated into bigraphs.
In $\bbig(\bsig)$, define 
the interfaces $I = ( 0, \emptyset, \widehat{\emptyset} )$
and $t = (1,\emptyset,\widehat{\emptyset})$.
\begin{lemma}\label{lemma:closed}
  The functor $\T$ induces an isomorphism $S (\theory_\bsig) (I, t)
  \iso \bbig (\bsig) (I, t)$.
\end{lemma}
\begin{proof}[sketch]
  For any morphism $I \rTo^{f} t$, we construct a candidate bigraph,
  and observe that it vacuously satisfies the binding rule.  Then, we
  proceed by contrapositive: assuming either that its parent map is
  cyclic or that it breaks the scope rule, we show that $f$ was
  incorrect.
\end{proof}

All in all, we have 
\begin{theorem}\label{thm:conservative}
  The functor $\T: \bbig(\bsig) \to S(\theory_\bsig)$ is faithful, essentially
  injective on objects, and surjective on $S(\theory_\bsig)(I, t)$.
\end{theorem}
It is however not full and far from surjective on objects.

\bibliographystyle{splncs}
\bibliography{bib}

\appendix

\section{Proof of Lemma~\ref{lemma:connected}}
Consider a switching of $\T(G)$.

Given a site or a node $p$, we denote by $\T(p)$ the 
negative $t$ port corresponding to it in the switching. If $p$ is
a root, then $\T(p)$ denotes the positive $t$ port of its
image.

Free ports of a node $p$ (resp. local names of a site $p'$) have
their image (a positive $v$ port) connected to $\T(p)$
(resp. $\T(p')$) as shown in Fig.~\ref{bigraph:connect_2}.
Moreover, either one negative $v$ port (corresponding to
a binding port) or the positive $t$ port of the cell $p$ is
connected to $\T(p)$ by the switched formula.
\begin{figure*}[th] \centering
  \includegraphics{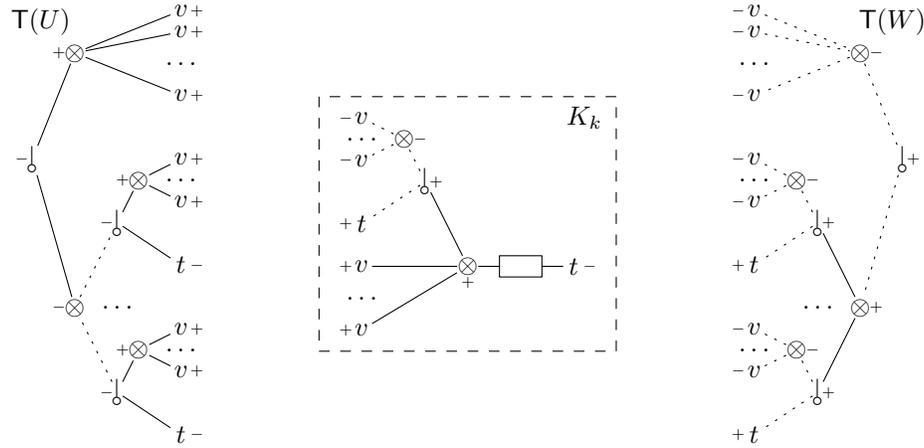}
  \caption{Domain, codomain and a node of a switching.}
  \label{bigraph:connect_2}
\end{figure*}

We now prove by induction that all binding ports (located at a
node or a root $p$) have their image connected to $\T(p)$.  Let
$b$ be a binding port, and $\T(b)$ its image by $\T$ (a
negative $v$ port).

If $b$ has no peers (this is necessarily the case if $p$ is a
barren node), then $\T(b)$ is connected to a $w$ cell whose $I$
port is connected to $\T(p)$.

If $b$ has peers, then $\T(b)$ is connected, in the morphism, to
their translations through a tree of $c$ cells.  But this tree is
heavily switched and only connects $\T(b)$ to one
positive $v$ port $f$ (whose preimage is) located, thanks to the
scope rule, to a site or a node $p' \prec p$.

By induction $f$ is connected to $\T(p')$ and $\T(p')$ is connected to
$\T(p)$ through the (unswitched) parent map. 
Indeed, the parent map connects the $t$ ports of cells between $p'$ and
$p$, and these cells have their $t$ ports connected thanks to the induction
hypothesis.
The port $\T(b)$ and $\T(p)$ are thus connected. 

Finally, we remark that:
\begin{itemize}
\item roots are connected to each other in the codomain's formula (by their $t$
  or $v$ ports, see Fig.~\ref{bigraph:connect_2}), 
\item global variables of the domain are connected to a site (by the domain's formula, see Fig.~\ref{bigraph:connect_2}) and
\item remaining negative $v$ ports (global variable
  of the codomain and $\nu$ cells) are connected to the other
  positive ports by a switched tree of $c$ cells or a
  $w$ cell.
\end{itemize}
We conclude that all ports of our switching are connected.

\section{Proof of Lemma~\ref{lemma:closed}}
  Consider any $f: I \to t$.  We have $\T(I) = (I \impll I) \iso I$
  and $\T(t) = I \impll (I \impll t) \iso t$, which justifies our
  ``induces'' above.  We now define $G = (V,E,\ctrl,\prnt,\link): I
  \to t$ such that $\T (G) = f$.

  Let the set of nodes $V$ be the set of logical cells in $f$; the
  control map $\ctrl$ sends each $K_k$ cell to $k \in \bsig$.  
  
  The set of edges is the coproduct of binding $v$ ports in the
  support of $f$ and of $\nu$ cells (where a $v$ port is binding when
  it occurs to the left of a $\impll$, e.g., a cell of type $((v \tens
  v) \impll t) \tens v \to t$ has two binding ports).

  The parent map $\prnt: 0 + V \to 1 + V$ is exactly the restriction
  of $f$ to $t$ ports.  The link map $\link: P_B + P_F + \emptyset \to
  E + \emptyset$ is obtained from the restriction of $f$ to $v$ ports
  as follows. From any $v$ port $p$, following the tree of
  contractions towards its root leads to a maximal positive
  $v$ port in the support, which may be either a port from a $\nu$
  cell, or a binding port of a logical cell.  In each case, there is a
  corresponding edge $e_p$.  Our link map sends each port $p$ to
  $e_p$. Since in each tree there is only one root, the binding rule
  is respected.

  We then prove that $G$ is correct.  Suppose that the parent map
  contains a cycle, then any switching where, for all cells of the
  cycle, the two $t$ ports are connected contains this cycle.  Suppose
  that the scope rule is not satisfied for a binder $p$ and one of its
  peers $p'$. Then, in $f$, $p$ is the root of a contraction tree with
  $p'$ as a leaf: among the switchings connecting them, choose again
  one that connects both $t$ ports of each logical cell: every logical
  cell then has a path to the root $r$ (the $t$ port in the codomain),
  which forms a cycle involving $p$, $p'$, and $r$, hence
  contradicting correctness of $f$.  The binding rule is automatically
  satisfied because the codomain has no name.  An atomic node has no
  antecedent in the parent map because the corresponding cell in $f$
  has no positive $t$ port.

\end{document}